\documentclass[journal,doublecolumn]{IEEEtran}
\usepackage{amsfonts}
\usepackage{amsmath}
\usepackage{amssymb}
\usepackage{subfig}
\usepackage{caption}
\usepackage{mathrsfs}
\usepackage{graphicx}
\usepackage{graphics}
\usepackage[compress]{cite}
\usepackage{times}
\usepackage{amsthm}
\usepackage{float}
\usepackage{footmisc}
\usepackage{xcolor}
\usepackage{float}
\usepackage{algorithm}
\usepackage{algpseudocode}

\newcommand{\R}{{\rm I\!R}}

\newtheorem{lemma}{Lemma}
\newtheorem{theorem}{Theorem}
\newtheorem{remark}{Remark}
\newtheorem{definition}{Definition}

\input{tcilatex}
\begin{document}

\title{Angle rigidity and its usage to stabilize\\ planar formations}
\author{Liangming Chen,  Ming Cao,   and Chuanjiang Li
	\thanks{%
	L. Chen and M. Cao are with Faculty of Science and Engineering, University of Groningen, Groningen, 9747 AG, The Netherlands. L. Chen and C. Li are with Department of Control Science and Engineering, Harbin Institute of Technology, Harbin, 150001, China. Email addresses: l.m.chen@rug.nl, m.cao@rug.nl, lichuan@hit.edu.cn.	
	
	} }
\maketitle

\begin{abstract}
	
Motivated by the challenging formation stabilization problem for mobile robotic teams when no distance or relative displacement measurements are available and each robot can only measure some of those angles formed by rays towards its neighbors, we develop the notion of ``angle rigidity" for a multi-point framework, named ``angularity", consisting of a set of nodes embedded in a Euclidean space and a set of angle constraints among them. Different from bearings or angles defined with respect to a global axis, the angles we use do not rely on the knowledge of a global coordinate system and are signed according to the counter-clockwise direction. Here \emph{angle rigidity} refers to the property specifying that under proper angle constraints, the angularity can only translate, rotate or scale as a whole when one or more of its nodes are perturbed locally. We first demonstrate that this angle rigidity property, in sharp comparison to bearing rigidity or other reported rigidity related to angles of frameworks in the literature, is \emph{not} a global property since an angle rigid angularity may allow flex ambiguity. We then construct necessary and sufficient conditions for \emph{infinitesimal} angle rigidity by checking the rank of an angularity's rigidity matrix. We develop a combinatorial necessary condition for infinitesimal minimal angle rigidity. Using the developed theories, a formation stabilization algorithm is designed for a robotic team to achieve a globally angle rigid formation, in which only angle measurements are needed.

\end{abstract}


\markboth{}{Shell \MakeLowercase{\textit{et al.}}: Bare Demo of
IEEEtran.cls for Journals}


%

\section{Introduction}

{ 
Over the past decades, \emph{distance rigidity} has been intensively investigated both as a mathematical topic in graph theory\cite{roth1981rigid,hendrickson1992conditions} and an engineering
problem in  applications including formations of multi-agent systems\cite{anderson2008rigid}, mechanical structures\cite{ildefonse1992mechanical}, and biological materials\cite{mayer2002rigid}. Distance rigidity\cite{asimow1979rigidity} is defined   using the property of distance preservation of translational and rotational motions of a multi-point framework. To determine whether a given framework is distance rigid, two methods have been reported. The first is to test the rank of the distance rigidity matrix which is derived from the infinitesimally distance rigid motions\cite{asimow1978rigidity}. The second is enabled by Laman's theorem, which is a  combinatorial test and works only for generic frameworks. More recently, \emph{bearing rigidity} has been investigated, in which the shape of a framework is prescribed by the inter-point bearings or directions\cite{eren2003sensor,zhao2016bearing}. By defining the bearing as an unit vector in a given global coordinate system, bearing rigidity can be defined accordingly \cite{eren2012formation,zhao2016bearing}. To check whether a framework is bearing rigid, the conditions similar to those for  distance rigidity have been discussed \cite{eren2003sensor,eren2012formation,bishop2015distributed,zhao2016bearing}.

Distance constraints in determining distance rigidity are in general  quadratic in the associated end points' positions. While a bearing constraint is always linear in the associated point's position, the description of bearings directly depends on the necessity of a global coordinate system or a coordinate system in $SE(2)$ or $SE(3)$\cite{zelazo2015bearing,michieletto2016bearing}. Different from distance and bearing rigidity, in this study we aim at presenting  \emph{angle rigidity} theory for multi-point frameworks accommodating angle constraints as either linear or quadratic constraints on the points' positions without the knowledge of a global coordinate system.  Different from the usual definition for a scalar angle, the angle defined in this paper is signed. By defining the counter-clockwise direction to be each angle's positive direction, {angle rigidity} is defined for an \emph{angularity} which consists of vertices and angle constraints. We show that  the planar angle rigidity is a local property because of the existence of flex ambiguity.  To check whether an angularity is angle rigid, angle rigidity matrix is derived based on the infinitesimally angle rigid motions. Then, the angle rigidity of an angularity can be determined by testing the rank of its angle rigidity matrix. Also, we develop a necessary combinatorial condition to test the angle rigidity of a generic angularity. We underline that the Laman's theorem and Henneberg's construction method do not apply directly to angle rigidity, which makes our results essential.

Besides its mathematical importance, angle rigidity is closely related to the  application in multi-agent formation control for robotic transportation\cite{li2008robust}, search and rescue of drones\cite{meng2014integrated}, and satellite formation flying in deep space\cite{kapila2000spacecraft}. Equipments used in formation stabilization mainly include Global Positioning System (GPS) receivers, radars, and cameras, which can acquire positions, inter-agent distances, or angles/bearings\cite{anderson2008rigid,oh2015survey}. In particular, angle measurements are becoming
cheaper, more reliable and accessible than relative position
or inter-agent distance measurements\cite{oh2015survey,zhao2019bearing}. Angle information can be easily
obtained by a vision-based camera in local coordinates\cite{das2002vision}. Using angle rigidity developed in this paper, we show how to stabilize a planar formation by using only angle measurements. Different from bearing-based control algorithms\cite{2zhao2017translational,zhao2016bearing} where all agents' local coordinate systems are required to be aligned, the proposed angle-based control algorithm does not require the alignment of agents'  coordinate systems since the description of an angle rigid angularity does not depends on the coordinate system. We acknowledge  that in \cite{jing2019angle}, planar angle rigidity is established by employing the cosine of an angle formed by two joint edges as the angle constraint. The formation stabilization algorithm constructed in \cite{jing2019angle} requires that each agent can sense the real-time relative displacements with respect to its neighbors. Different from \cite{jing2019angle}, in this paper the desired formation shape is realized using only angle measurements. In addition, weak rigidity with mixed distance and angle constraints has been investigated in \cite{park2017rigidity,kwon2018infinitesimal,kwon2018infinitesimal1}, under which the formation control algorithms  are also designed for agents by using the measurements of relative displacement.

The rest of this paper is organized as follows. Section
II gives the definition of an angularity and its rigidity. Section
III introduces generic and infinitesimal angle rigidity. In Section IV, the application in multi-agent planar formations is investigated.

}

%
%
%

\section{Angularity and its rigidity}

Graphs have been used dominantly in rigidity theory for multi-point frameworks under distance constraints since an edge of a graph can be used naturally to denote the existence of a distance constraint between the two points corresponding to the two vertices adjacent to this edge. However, when describing angles formed by rays connecting points, to use edges of a graph becomes clumsy and even illogical because  an angle constraint always involves three points. For this reason, instead of using graphs that relate pairs of vertices as the main tool to define rigidity, we define a new combinatorial structure ``angularity" that relates triples of vertices to develop the theory of angle rigidity. In all the following discussions we confine ourselves to the plane.

\subsection{Angularity}
We use the vertex set $\mathcal{V}=\{1,2,\cdots,N\}$  to denote the set of indices of the  $N\geq 3$ points of a framework in the plane. As shown in Fig. \ref{AR/fig-angle}, to describe the \emph{signed} angle from the ray $j$-$i$ to ray $j$-$k$, one needs to use the ordered triplet $(i,j,k)$, and obviously the two angles corresponding to $(i,j,k)$ and $(k,j,i)$ are different, and in fact are called explementary or conjugate angles. Here, following convention, the angle $\measuredangle ijk$ for each  triplet $(i,j,k)$ is measured counterclockwise in the range  $[0,2\pi)$. We use  $\mathcal{A}\subset \mathcal{V} \times  \mathcal{V} \times  \mathcal{V}=\{(i,j,k), \forall i,j,k \in \mathcal{V},i\neq j \neq k \}$ to denote the \emph{angle set}, each element of which is an ordered triplet. We denote the number of elements $|\mathcal{A}|$ of the  angle set $\mathcal{A}$ by $M$. {Throughout this paper, we assume that no pair of triplets in $\mathcal{A}$ are explementary to each other}. Now consider the embedding of the vertex set $\mathcal V$ in the plane $\R ^2$ through which each vertex $i$ is associated with a distinct position  $p_i \in \R^2$ and let $p=[p_1^T,\cdots,p_N^T]^T \in \R^{2N}$. We assume the positions do not coincide. Then the combination of the vertex set $\mathcal V$, the angle set $\mathcal A$ and the position vector $p$ is called an \emph{angularity}, which we denote by $\mathbb A (\mathcal V, \mathcal A, p)$.

\begin{figure}[ht]
	\centering
	\includegraphics[width=5cm]{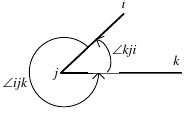}
	\caption{Angle used in defining angle rigidity.}
	\label{AR/fig-angle}
	\centering
\end{figure}

\subsection{Angle rigidity}
We first define what we mean by two equivalent or congruent angularities. 	
\begin{definition}\label{AR/def-equivalent}
We say two angularities  $\mathbb{A}(\mathcal V, \mathcal A,p)$ and $\mathbb{A}(\mathcal V, \mathcal A, p')$ with the same $\mathcal{V}$ and $\mathcal{A}$ are \emph{equivalent} if
\begin{equation}
\measuredangle ijk (p_i, p_j, p_k)= \measuredangle ijk(p_i', p_j', p_k') \textrm{\; for\ all\;} (i,j,k)\in \mathcal A.
\end{equation}
We say they are \emph{congruent} if
\begin{equation}
\measuredangle ijk (p_i, p_j, p_k)= \measuredangle ijk(p_i', p_j', p_k') \textrm{\; for\ all\;} i,j,k\in \mathcal V.
\end{equation}
\end{definition}

From the equivalent and congruent relationships, it is easy to define global angle rigidity.
\begin{definition}\label{AR/def-globalrigidity}
	An angularity $\mathbb{A}(\mathcal{V}, \mathcal A,p)$ is \emph{globally angle rigid} if every angularity that is equivalent to it is also  congruent to it.	
\end{definition}
When such a rigidity property holds only locally, one has angle rigidity.

\begin{definition}\label{AR/def-localrigidity}
An angularity $\mathbb{A}(\mathcal{V},\mathcal{A},p)$  is \textit{{{angle rigid}}} if there exists an $\epsilon>0$ such that every angularity $\mathbb{A}(\mathcal{V}, \mathcal{A}, p')$ that is  equivalent to it and satisfies $\left\| {p}'-p \right\|<\epsilon$, is congruent to it.
\end{definition}

Definition \ref{AR/def-localrigidity} implies that every configuration which is sufficiently close to $p$ and satisfies all the angle constraints formed by 
$\mathcal{A}$, has the same magnitudes of the angles formed by any three vertices in
$\mathcal{V}$ as the original configuration at $p$.

As is clear from Definitions \ref{AR/def-globalrigidity} and \ref{AR/def-localrigidity}, global angle rigidity always implies angle rigidity. A natural question to ask is whether angle rigidity also implies global angle rigidity. In fact, for bearing rigidity, it has been shown that indeed global bearing rigidity and bearing rigidity are equivalent \cite{eren2012formation,zhao2016bearing}. However, this is \emph{not} the case for angle rigidity. 

\begin{theorem}
An angle rigid angularity  $\mathbb{A}(\mathcal{V},\mathcal{A},p)$ is not necessarily globally angle rigid.
\end{theorem}
We prove this theorem by providing the following example.\begin{figure}[H]
	\centering
	\includegraphics[width=6.0cm]{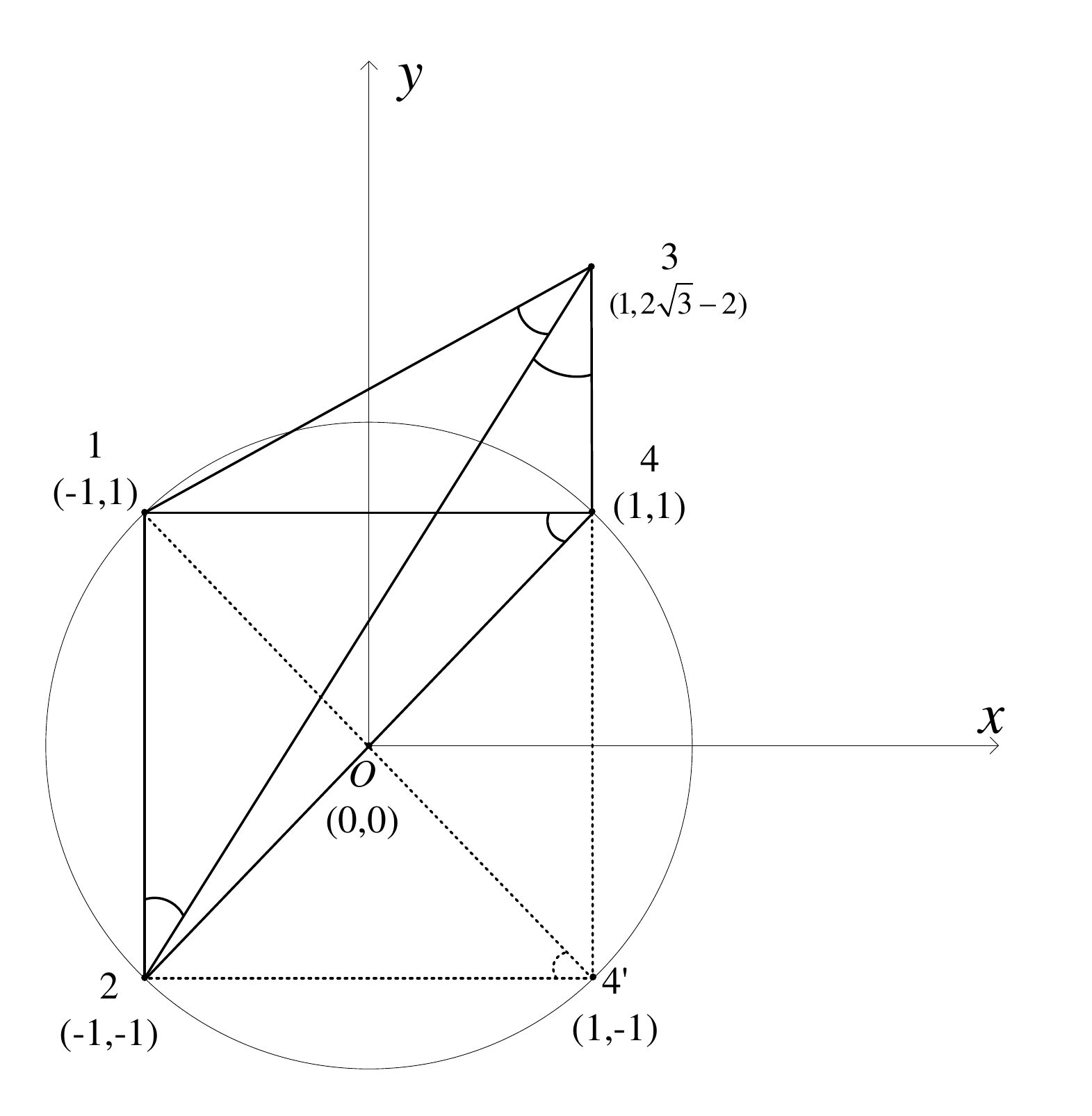}
	\caption{Flex ambiguity in angle rigid angularity}
	\label{AR/fig-ambiguity}
\end{figure}
\noindent Fig. \ref{AR/fig-ambiguity} shows an angularity with $\mathcal V = \{1,2,3,4\}$,  its elements in the set $\mathcal A = \{(3,2,1), (1,3,2), (2,3,4), (1,4,2)\}$ taking the values  
\begin{equation}
\measuredangle 321=\arccos(\frac{4\sqrt{3}-2}{2\sqrt{17-4\sqrt{3}}})\approx39.07^{\circ},
\end{equation}
\begin{equation}
 \measuredangle 132=\arccos(\frac{19-8\sqrt{3}}{\sqrt{25-12\sqrt{3}}\sqrt{17-4\sqrt{3}}})\approx37.88^{\circ},
\end{equation}
\begin{equation}
\measuredangle 234=30^{\circ},
\end{equation}
\begin{equation}
\measuredangle 142=45 ^{\circ},
\end{equation}
and its $p$ is shown as in the coordinates of the vertices. Now first look at the triangle formed by 1, 2 and 3. Since two of its angles $\measuredangle 321$ and $\measuredangle 132$ have been constrained, the remaining  $\measuredangle 213$ is uniquely determined to be $\pi - \measuredangle 321- \measuredangle 132$. The constraint on $\measuredangle 234$ requires 4 must lie in the ray starting from 3 and rotating from the ray 32 anticlockwise by 30 degree; at the same time, the constraint on $\measuredangle 142$ requires 4 must lie on the circle passing through 1 and 2  such that the inscribed angle $\measuredangle 142$ is 45 degree. If we fix the positions of 1, 2, and 3, then there is only one unique position for 4 in the neighborhood of its current given coordinates as the intersection point of the ray and the circle. This local uniqueness implies that this four-vertex angularity is angle rigid (when 4's position is uniquely determined, any angle associated with it is also uniquely determined); however, globally, there is the other intersection point $4'$ as shown in the figure, which implies that this angularity is not globally angle rigid.  \hfill $\square$

We provide the following further insight to explain this sharp difference between the angle rigidity that we have defined and the bearing rigidity that has been reported in the literature.  Bearing rigidity as defined in \cite{eren2012formation,zhao2016bearing} is a global property because the bearing constraints are always linear in $p$ when written as a linear  constraint (similar to the constraint in the form of the ray from 3 to 4 in the example) in some global coordinate system. In contrast, our angle constraints can be either linear in $p$ when it requires the corresponding vertex to be on a ray or quadratic in $p$ when it restricts the corresponding vertex to be on an arc passing through other vertices. The possible nonlinearity in the angle constraints gives rise to potential ambiguity of the vertices' positions.

Note that the embedding of $p$ in the plane may affect the rigidity of $\mathbb A$. Consider the 3-vertex angularity as embedded in the following three different situations when its angle set $\mathcal A$ contains only one element $(2,1,3)$. \begin{figure}[H]
	\centering
	\includegraphics[width=9.0cm]{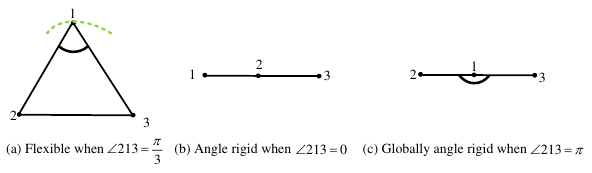}
	\caption{Non-generic $p$ changes rigidity}
	\label{AR/fig-specialgeneric}
\end{figure}
\noindent Sub-figure (a) shows that 1, 2, 3 are not collinear,  and then this angularity is in general flexible since if we fix the positions of 2 and 3, then the constraint on $\measuredangle 213$ still allows 1 to move along an arc and correspondingly the angles $\measuredangle 123$ and $\measuredangle 132$ change. In sub-figure (b), 1, 2, 3 are collinear and 1 is on one side, in this case if the angle constraint happens to be $\measuredangle 213 =0$, then one can check the angularity becomes angle rigid, although it is not globally rigid since the angle of $\measuredangle 132$ changes by 180 degree if we swap 1 and 3. In the last sub-figure (c), 1, 2, 3 are collinear and 1 is in the middle, when the constraint becomes $\measuredangle 213 = \pi$, one can check that the angularity is not only rigid, but also globally rigid (swapping of 2 and 3 in this case does not change the resulting angles being zero). So the angularity $\mathbb A (\{1,2,3\}, \{(2,1,3)\},p)$ is generically flexible, but rarely rigid depending on $p$. To clearly describe this relationship between angle rigidity and $p$, like in  standard rigidity theory, we define what we mean by generic positions. 

\begin{definition}
	The position vector $p$ is said to be \emph{generic} if its components are algebraically independent \cite{connelly2005generic}. Then we say an angularity is \emph{generically (globally) angle rigid} if its $p$ is generic and it is (globally) angle rigid. 
\end{definition}

For convenience, we also say an angularity is generic if its $p$ is generic. Now we provide some sufficient conditions for an angularity  to be globally angle rigid. Towards this end, we need to introduce some concepts and operations. For two angularities $\mathbb A (\mathcal V, \mathcal A, p )$ and $\mathbb A' (\mathcal V', \mathcal A', p')$, we say $\mathbb A$ is a \emph{sub-angularity} of $\mathbb A'$ if $\mathcal V \subset \mathcal V'$, $\mathcal A \subset \mathcal A'$ and $p$ is the corresponding sub-vector of $p'$. We first clarify that for the smallest angularities, namely those contains only three vertices, there is no gap between global and local generic angle rigidity.
\begin{lemma}	\label{AR/lem-3vertexangularity}
	For a 3-vertex angularity, if it is generically angle rigid, it is also generically globally angle rigid.
\end{lemma}
\begin{proof}
	For this 3-vertex angularity $\mathbb A (\mathcal V, \mathcal A, p )$, since it is angle rigid and  $p$ is generic, $\mathcal A$ must contain at least two elements, or said differently, two of the interior angles of the triangle formed by the three vertices are constrained. Again since $p$ is generic, the sum of the three interior angles in this triangle has to be $\pi$, and thus the magnitude of this triangle's remaining interior angle is uniquely determined too. Therefore, $\mathbb A$ is generically globally angle rigid.
\end{proof}

Now, we define the vertex addition operations and the aim is to demonstrate how a bigger angularity might grow from a smaller one.
\begin{definition}
	For a given angularity $\mathbb A (\mathcal V, \mathcal A, p )$,   a new vertex $i$ positioned at $p_i$ is \emph{linearly constrained} with respect to $\mathbb A$ if there is $ j\in \mathcal{V}$ such that $p_i \neq p_j$ and $p_j$ is constrained to be on a ray starting from $p_j$; we also say $i$ is \emph{quadratically constrained} with respect to $\mathbb A$ if there are $j, k\in \mathcal V$ such that $\{p_i, p_j, p_k\}$ is generic and $p_i$ is constrained to be on an arc with $p_j$ and $p_k$ being the arc's two ending points. Correspondingly, we call $i$'s constraint in the former case a  \emph{linear constraint} and in the latter case a \emph{quadratic constraint}   with respect to $\mathbb{A}$.
\end{definition}

\begin{definition}[Type-I vertex addition]\label{Type-Ivertexaddition}
	For a given angularity $\mathbb A (\mathcal V, \mathcal A, p )$, we say the angularity $\mathbb A'$ with the augmented vertex set $\{\mathcal V \cup \{i\}\}$ is obtained from $\mathbb A$ through a \emph{Type-I vertex addition} if the new vertex $i$'s constraints with respect to $\mathbb A$ contain at least one of the following: 
		
	1) two linear constraints, not aligned, associated with two distinct vertices in $\mathcal V$ (one vertex for one constraint and the other vertex for the other constraint); 
	
	2) one linear constraint and one  quadratic constraint associated with two distinct vertices in $\mathcal V$ (one for the former and both for the latter); 
	
	3) two quadratic constraints associated with three vertices in $\mathcal V$ (two for each and one is shared by both).
\end{definition}

\begin{definition}[Type-II vertex addition]\label{Type-IIvertexaddition}
	For a given angularity $\mathbb A (\mathcal V, \mathcal A, p )$,  we say the angularity $\mathbb A'$ with the augmented vertex set $\{\mathcal V \cup \{i\}\}$ is obtained from $\mathbb A$ through a \emph{Type-II vertex addition} if the new vertex $i$'s constraints with respect to $\mathbb A$ contain at least one of the following: 
	
	1) one linear constraint and one quadratic constraint associated with three distinct vertices in $\mathcal V$ (one for the former and the other two for the latter); 
	
	2) two different quadratic constraints associated with four vertices in $\mathcal V$ (two for the former and the other two for the latter).
\end{definition}
	\begin{figure}[H]
	\centering
	\includegraphics[width=8.7cm]{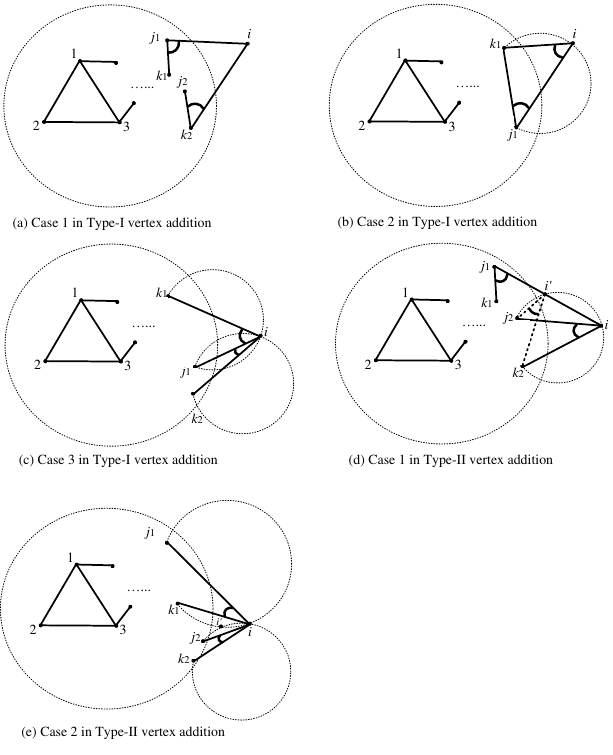}
	\caption{Type-I vertex addition and Type-II vertex addition}
	\label{AR/fig-vertexaddition}
\end{figure}

\begin{remark}
	The numbers of vertices involved in condition (2)  in Definition \ref{Type-Ivertexaddition} and condition (1) in Definition  \ref{Type-IIvertexaddition} differ in these two types of vertex addition operations. Similarly, those in   condition (3)  in Definition \ref{Type-Ivertexaddition} and condition (2) in Definition  \ref{Type-IIvertexaddition} are also different.
\end{remark}

\begin{remark}
	Note that in these two vertex addition operations, all the involved vertices are required to be in generic positions. However, the overall angle rigid angularity $\mathbb A'$ constructed through a sequence of vertex addition operations is \emph{not} necessarily generic, and an example is given in Fig. \ref{AR/fig-vertexadditiongeneric}.
	\begin{figure}[H]
		\centering
		\includegraphics[width=8.7cm]{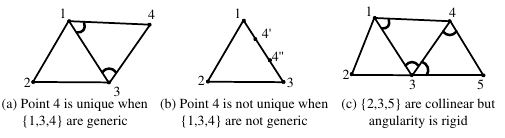}
		\caption{The overall angularity is not necessarily generic}
		\label{AR/fig-vertexadditiongeneric}
	\end{figure}
\end{remark}

Now we are ready to present a sufficient condition for global angle rigidity using type-I vertex addition.
\begin{proposition} \label{AR/prop-globalrigid}
	An angularity is {{{globally angle rigid}}} if it can be obtained through a sequence of Type-I vertex additions from a generically angle rigid 3-vertex angularity.
\end{proposition}
\begin{proof}
	According to Lemma \ref{AR/lem-3vertexangularity}, the generically angle rigid 3-vertex angularity is  globally angle rigid. Consider the three conditions in the Type-I vertex addition. If 1) applies, then  the position $p_i$ of the newly added vertex $i$ is unique since two rays, not aligned, starting from two different points may intersect only at one point; if 2) applies, $p_i$ is again unique since a ray starting from the end point of an arc may intersect with the arc at most at one other point; and if 3) applies, $p_i$ is unique since two arc sharing one end point on different circles can only intersect at most at one other point. Therefore, $p_i$ is always globally uniquely determined. In addition, the positions of the vertices after a sequence of type-I vertex additions are not necessarily generic, so we conclude that the obtained angularity is globally angle rigid.
\end{proof}

In comparison, type-II vertex additions can only guarantee angle rigidity, but not global angle rigidity. 
\begin{proposition} \label{AR/prop-rigid}
	An angularity is angle rigid if it can be obtained through a sequence of Type-II vertex additions from a generically angle rigid 3-vertex angularity.
\end{proposition}
The proof can be easily constructed following similar arguments as those for Proposition \ref{AR/prop-globalrigid}. The only difference is that $p_i$ now may have two solutions and is only unique locally.

{

}

After having presented our results on angularity and generic angle rigidity, in the following sectoin, we discuss infinitesimal angle rigidity, which relates closely to infinitesimal motion.  

\section{Generic and infinitesimal angle rigidity}

Analogous to distance rigidity, infinitesimal angle rigidity can be characterized by the kernel of a properly defined rigidity matrix. Towards this end, we first introduce the following angle function. For each angularity $\mathbb A (\mathcal V, \mathcal A, p)$, we define the \textit{{{angle function}}} $f_{\mathcal{A}}(p): \R ^{2N}\to \R^M$ by
\begin{equation}
f_{\mathcal{A}}(p):=[f_1,\cdots,f_M]^T,
\end{equation}
where $f_m:\R ^6\to [0,2\pi)$,  $m=1,\cdots,M$, is the mapping from the position vector $[p_i^T, p_j^T, p_k^T]^T$ of the $m$th element $(i,j,k)$ in $\mathcal A$   to the signed angle $\measuredangle ijk$. Using this angle function, one can define $\mathbb A$'s angle rigidity matrix.

\subsection{Angle rigidity matrix}

Following \cite{eren2003sensor}, we consider an arbitrary  element $(i,j,k)$ in  $\mathbb A$ and denote the corresponding angle constraint by $\measuredangle ijk (p_i, p_j, p_k)=\beta$, or in shorthand $\measuredangle ijk=\beta$, where $\beta \in [0,2\pi)$ is a constant. 
From the definition of the dot product,  one has
\begin{equation}
\|p_i -p_j \|\; \|p_k -p_j \|\cos\beta=(p_i -p_j) ^T (p_k -p_j), \label{AR/eq-ancon1}
\end{equation}
where $\left\Vert \cdot
\right\Vert $ denotes the Euclidean vector norm and we have used the fact that $\cos\beta=\cos(2\pi-\beta)$.  Taking the square of both sides and then differentiating with respect to time lead to
\begin{align} \label{AR/eq-ancon}
&\big (l^{2}_{jk}(p_i-p_j)\cdot (\dot p_i-\dot p_j)+l^{2}_{ji}(p_k-p_j)\cdot(\dot p_k-\dot p_j)\big)\cos \beta  \notag \\
&=l_{jk}l_{ji}\{(p_k-p_j)\cdot (\dot p_i-\dot p_j)+(p_i-p_j)\cdot (\dot p_k-\dot p_j)\},
\end{align}
where $l_{jk} = \|p_j - p_k\|$ and $l_{ji}= \|p_j - p_i\|$. 
Dividing both sides by $l_{jk}l_{ji}$ and  rearranging terms, one obtains
\begin{equation} \label{eq/ABC}
A\cdot \dot p_i+B\cdot \dot p_j+C\cdot \dot p_k=0,
\end{equation} 
where 
\begin{align}
A&=\frac{(p_i-p_j)^{\bot}}{l_{ij}}l_{jk}\sin\beta,\\
B&=-\frac{(p_i-p_j)^{\bot}}{l_{ij}}l_{jk}\sin\beta+\frac{(p_k-p_j)^{\bot}}{l_{kj}}l_{ij}\sin\beta, \\
C&=\frac{(p_j-p_k)^{\bot}}{l_{kj}}l_{ij}\sin\beta,
\end{align}
and for a vector $p$, $p^{\bot}$ is the vector obtained by rotating $p$ counterclockwise by $\frac{\pi}{2}$.
For each $(i,j,k)$ in $\mathcal A$ we obtain an equation in the form of (\ref{eq/ABC}), and then one can write such $M$ equations into the matrix form $B(p)\dot p=0$ where 
$B(p)\in \R^{M\times 2N}$ is called the \emph{angle rigidity matrix}, whose rows are indexed by the  elements of $\mathcal A$ and columns the coordinates of the vertices:

\begin{equation}
B(p)=\notag \qquad \qquad \qquad
\end{equation}
\begin{small}
\begin{equation}
\begin{bmatrix}
 &\cdots  & \text{Vertex}\ i &\cdots &  \text{Vertex}\ j &\cdots & \text{Vertex}\ k & \cdots \cr
\text{Angle}\ 1  &     \cdots &  \cdots &  \cdots & \cdots  &\cdots &\cdots  &\cdots \cr
\cdots &  \cdots &  \cdots &  \cdots & \cdots &\cdots &\cdots  &\cdots \cr
\measuredangle ijk &  0& N_{ij} &0 & N_{ji}+N_{kj}&0  &  N_{jk} &0 \cr
\cdots &  \cdots &  \cdots &  \cdots & \cdots &\cdots &\cdots  &\cdots \cr
\text{Angle} \ M  & \cdots &  \cdots &  \cdots &\cdots &\cdots &\cdots  &\cdots \cr
\end{bmatrix}
\label{AR/eq-eq11}
\end{equation}\end{small}

and \begin{equation}
N_{ij}=\left( \frac{(p_i-p_j)^{\bot}}{l_{ij}^2}\right)^T. \label{AR/eq-eq8}
\end{equation}

Since for an angularity, its angle preservation motions include translation, rotation, and scaling, one may rightfully expect that such motions are captured by the null space of the angle rigidity matrix, which always contains the following four linearly independent vectors 
\begin{equation}
q_1=1_{N}\otimes
\begin{bmatrix}
1\\ 0
\end{bmatrix},
\end{equation}

\begin{equation}
q_2=1_{N}\otimes
\begin{bmatrix}
0\\ 1
\end{bmatrix},
\end{equation}

\begin{equation}
q_3=
\begin{bmatrix}
(Q_0p_1^{\bot})^{T} ,& (Q_0p_2^{\bot})^{T}, & \cdots, & (Q_0p_N^{\bot})^{T}
\end{bmatrix}^T,
\end{equation}

\begin{equation}
q_4=
\begin{bmatrix}
(\alpha p_1)^{T} ,& (\alpha p_2)^{T}, & \cdots, & (\alpha p_N)^{T}
\end{bmatrix}^T,
\end{equation}
where $Q_0=\begin{bmatrix} 0 & 1 \\-1 & 0 \end{bmatrix}$ is skew symmetric,  $\alpha\in \R$ is a constant scaling factor, and $\otimes$ represents Kronecker product. Note that $q_1$ and $q_2$ correspond to translation, $q_3$ rotation, and $q_4$ scaling.
We state this fact as a lemma.

\begin{lemma} \label{AR/lem-rankrigiditymatrix}
For an angle rigidity matrix $B(p)$, it always holds that	 $\text{Span}\{q_1,q_2,q_3,q_4\}\subseteq \text{Null}(B(p))$ and correspondingly	 $\text{Rank}(B(p))\leq 2N-4$.
\end{lemma}

Obviously the row rank of the angle rigidity matrix, or equivalently its row linear dependency, is a critical property of an angularity. We capture this property by using the notion of ``independent" angles. 

\begin{definition}\label{AR/def-dependence}
	For an angularity $\mathbb{A}(\mathcal{V},\mathcal{A},p)$, we say its angles in $f_\mathcal{A}(p)$   are \emph{independent} if its angle rigidity matrix $B(p)$ has full row rank. 
\end{definition}

Since rank is a generic property of a matrix, one may wonder whether it is possible to disregard $p$ of $\mathbb A$ and define angle rigidity only using $\mathcal A$. This is indeed doable as what we will show in the following subsection. 
Note that $2N-4$ is the maximum rank that $B(p)$ can have. When $p$ is generic, the exact realization of $p$ is not important, and when checking the angle rigidity matrix's rank, one can replace $p$ by a random realization. 

Using the notion of infinitesimal motion, checking the rank of the rigidity matrix can also enable us to check  ``infinitesimal" angle rigidity.

\subsection{Infinitesimal angle rigidity}
To consider infinitesimal motion, suppose that each $ p_i, \forall i\in \mathcal{V}$ of $ \mathbb{A}(\mathcal{V},\mathcal{A},p)$ is on a differentiable smooth path. We say the whole path $p(t)$ is generated by an \textit{{{infinitesimally angle rigid motion}}} of $\mathbb A$ if on the path $f_\mathcal{A}(p)$ remains constant. We say such an infinitesimally angle rigid motion $p(t)$ is \textit{{{trivial}}} if it can be given by \cite{connelly2015frameworks}
\begin{equation}
p_i(t)=\alpha(t) Q(t)p_i(t_0)+W(t),\forall i\in \mathcal{V}, t\geq t_0, \label{AR/eq-eq13}
\end{equation}
where $\alpha(t)\neq 0$ is a scalar scaling factor, $Q(t)\in \R^{2\times 2}$ is a rotation matrix, $W(t)\in \R^{2}$ is a translation vector, and $\alpha(t),Q(t),W(t)$ are all differentiable smooth functions. Since all  $p_i(t),\forall i \in \mathcal{V}$, share the same $\alpha(t), Q(t), W(t)$, it follows
\begin{equation}
p(t)=\{I_{N}\otimes  [\alpha(t)Q(t)]\}p(t_0)+1_{N}\otimes W(t), t\geq t_0. \label{AR/eq-eqn12}
\end{equation}
where $I_N$ and $1_N$ denote the $N\times N$ identity matrix and $N\times 1$ column vector of all ones, respectively. Now we are ready to define infinitesimal angle rigidity.
\begin{definition}\label{AR/def-infinitesimal2}
	An angularity $\mathbb{A}(\mathcal{V},\mathcal{A},p)$ is \textit{{{infinitesimally angle rigid}}} if all its continuous infinitesimally angle rigid motion $p(t)$ are trivial.
\end{definition}
In fact, if the motion $p(t)$ always satisfy (\ref{eq/ABC}), it must be a combination of translation, rotation and scaling of $\mathbb A$, which must be a motion in  (\ref{AR/eq-eqn12}). The converse also holds, namely a trivial motion satisfying (\ref{AR/eq-eqn12}) is always a combination of translation, rotation and scaling and thus preserves angle constraints as indicated by (\ref{eq/ABC}). We formalize these remarks in the following theorem. 


\begin{theorem}  \label{AR/theo-infinitesimalrigidity}
	An angularity $ \mathbb{A}(\mathcal{V},\mathcal{A},p)$ is infinitesimally angle rigid if and only if the rank of its angle rigidity matrix $B(p)$ is $2N-4$. 
\end{theorem}

\begin{proof} 	In view of the definition,  $\mathbb{A}$ is infinitesimally angle rigid if and only if all its infinitesimally angle rigid motions are trivial. That is to say, these trivial infinitesimally angle rigid motions $p(t),t\in[t_0,t_1]$ are exactly the combination of translation, rotation, and scaling with respect to the initial configuration $p(t_0)$, which are precisely captured by the four linearly independent vectors $q_1$, $q_2$, $q_3$, and $q_4$, which in turn is equivalent to the fact that the rigidity matrix's null space is precisely the span of $\{q_1, q_2, q_3, q_4\}$. The conclusion then follows from the fact that such a specification of the null space holds if and only if the rank of the rigidity matrix reaches its maximum $2N-4$.
\end{proof}

Note that this theorem implies that  $ \mathbb{A}(\mathcal{V},\mathcal{A},p)$ is infinitesimally angle rigid if and only if there are $2N-4$ independent angles in $f_\mathcal A(p)$. We want to further remark that when $p$ is generic, and if one of the following three combinatorial structures appears, then the angles are always dependent.  \\
(1) A cycle formed by the triplets in $\mathcal A$. For example, $\mathcal{A}=\{(i,j,k),(j,k,m),(k,m,n),(m,n,l),(n,l,i),(l,i,j)\}$, see Fig. \ref{AR/fig-dependence}.(a).\\
(2) Angles around a vertex. For example, $\mathcal{A}=\{(i,m,j),(j,m,k),(k,m,i)\}$, see Fig. \ref{AR/fig-dependence}.(b).\\
(3) A nonempty subset $\mathcal{A}'\subset \mathcal{A}$ such that the number  $N'$ of the involved vertices in $\mathcal{A}'$ satisfies $|\mathcal{A}'|> 2 N'-4$. For example, $\mathcal{A}=\{(i,m,j),(m,j,i),(i,k,j),(i,j,k),(k,m,j),(n,i,m),\\(n,m,i)\}$ and $\mathcal{A}'=\{(i,m,j),(m,j,i),(i,k,j),(i,j,k),\\(k,m,j)\}$, and thus $N'=4$, $|\mathcal{A}'|=5$ in Fig. \ref{AR/fig-dependence}. (c).

\begin{figure}[H]
	\centering
	\includegraphics[width=7cm]{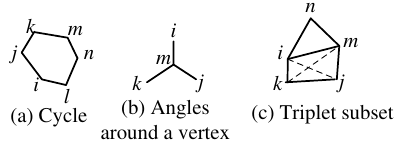}
	\caption{Types of dependent triplet elements}
		\label{AR/fig-dependence}
\end{figure}
If $\mathcal{A}$ contains one of the above three combinatorial structures, we say the triplet elements in $\mathcal{A}$ are dependent; otherwise, they are independent. One can further quantify the number of triplet elements such that the angularity is infinitesimally angle rigid. 
\begin{theorem} \label{AR/theo-infinitesimaltriplet}
	For an angularity $ \mathbb{A}(\mathcal{V},\mathcal{A},p)$, if it is infinitesimally angle rigid, then it has $2N-4$ independent triplet elements in $\mathcal{A}$.
\end{theorem}

\begin{proof}
	From Theorem \ref{AR/theo-infinitesimalrigidity}, we know $\mathbb{A}$ has $2N-4$ independent angles in $f_{\mathcal{A}}(p)$. In addition, by using the structure of angle rigidity matrix $B(p)$, it is easy to prove that dependent triplet elements in $\mathcal{A}$ $\Rightarrow$ dependent angles in $f_\mathcal{A}(p)$, which implies that independent angles in $f_\mathcal{A}(p)$ $\Rightarrow$ independent triplet elements in $\mathcal{A}$. So its angle set $\mathcal{A}$ has $2N-4$ independent triplet elements. 
\end{proof}

Now we show that when $p$ is generic, angle rigidity and infinitesimal angle rigidity are equivalent. For an angularity $\mathbb{A}(\mathcal{V},\mathcal{A},p)$ with a given $p$, define $\mathcal{M} (\mathbb A):=\{q\in \R^{2N}| \mathbb{A}(\mathcal{V},\mathcal{A},q)\text{ is congruent to~} \mathbb{A}(\mathcal{V},\mathcal{A},p)\}$, which is the manifold where the angle functions $f_{\mathcal{A}^*}(q)$ remain the same as $f_{\mathcal{A}^*}(p)$ where $\mathcal{A}^*=\mathcal{V} \times  \mathcal{V} \times  \mathcal{V}=\{(i,j,k), \forall i,j,k \in \mathcal{V},i\neq j\neq k\}$.

	\begin{lemma}\label{AR/lem-genericcoincide}
	An angularity $\mathbb{A}(\mathcal{V},\mathcal{A},p)$ is angle rigid if and only if $\mathcal{M}$ and $f_{\mathcal{A}}^{-1}(f_{\mathcal{A}}(p))$ coincide near $p$.
\end{lemma}
The proof is similar to that for distance rigidity in \cite[Proposition 5.1]{roth1981rigid} and thus omitted here.

\begin{theorem} \label{AR/theo-genericinfinitesimal}
When $p$ is generic, an angularity $ \mathbb{A}(\mathcal{V},\mathcal{A},p)$ is infinitesimally angle rigid if and only if it is angle rigid.
\end{theorem}

\begin{proof} 	
(Sufficiency) Since  $\mathcal{M}$ is a subset of $f_{\mathcal{A}}^{-1}(f_{\mathcal{A}}(p))$ as $\mathcal{A}$ is a subset of $\mathcal A^*$, when $\mathbb A$ is infinitesimally angle rigid,   $\mathcal{M}$  becomes the 4-dimensional manifold of configurations corresponds to the trivial infinitesimally angle rigid motions. From Lemma \ref{AR/lem-genericcoincide} we know that $\mathcal{M}$ and $f_{\mathcal{A}}^{-1}(f_{\mathcal{A}}(p))$ coincide near $p$, so the motions from $p$ to $f_{\mathcal{A}}^{-1}(f_{\mathcal{A}}(p))$ are always trivial when $p$ is generic. Then $\mathbb A$ is infinitesimally angle rigid according to Definition \ref{AR/def-infinitesimal2}.
	
(Necessity)	From Definition \ref{AR/def-infinitesimal2}, we know that all the continuous infinitesimally angle rigid motion $p(t)$ are trivial, which are the combination of  translation, rotation, and scaling of $\mathbb{A}(\mathcal{V},\mathcal{A},p)$. Consider another angularity $\mathbb{A}(\mathcal{V}, \mathcal{A}, p')$ with $\varepsilon>0$ and $\|p'-p\|<\varepsilon$, which is equivalent to $\mathbb{A}(\mathcal{V}, \mathcal{A}, p)$. Then, the continuous motion from $p$ to $p'$ are the combination of translation, rotation and scaling of $\mathbb{A}(\mathcal{V},\mathcal{A},p)$, which are angle-preserving motion, i.e., $f_{\mathcal{A}^*}(p)$ remain constant. Therefore,  $\mathbb{A}(\mathcal{V}, \mathcal{A}, p')$ is congruent to  $\mathbb{A}(\mathcal{V}, \mathcal{A}, p)$, which implies that $\mathbb{A}(\mathcal{V}, \mathcal{A}, p)$ is angle rigid. 
\end{proof}

We use the following example to illustrate the difference between angle rigidity and infinitesimal angle rigidity. The angularity in the left of Fig. \ref{AR/fig-generic}  is angle rigid but not infinitesimally angle rigid, while the angularity on the right is both angle rigid and infinitesimally angle rigid. 

\begin{figure}[H]
	\centering
	\includegraphics[width=7.8cm]{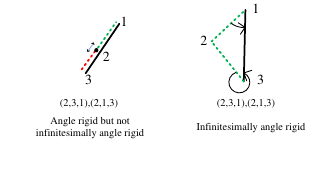}
	\caption{Difference between angle rigid angularity and infinitesimally angle rigid angularity}
	\label{AR/fig-generic}
\end{figure}

We further use the following examples to illustrate the difference among independent triplet elements, generic configuration, and infinitesimal angle rigidity, where the angularities in (a) and (b) share the same shape and the angularities in (b) and (c) share the same angle set $\mathcal A$. The angularity in Fig. \ref{AR/fig-relation}(a) is angle rigid although $p_2$, $p_3$, and $p_4$ are collinear; the one in (b) is angle flexible as it admits another positioning of 2 and 3 at $p_2'$ and $p_3'$ respectively. This is because the three collinear points exactly distributed in the two triplet elements $(3,2,4)$ and $(2,4,3)$.  The one in  Fig. \ref{AR/fig-relation} (c) is infinitesimally angle rigid, and thus equivalently generically angle rigid.
	
\begin{figure}[H]
	\centering
	\includegraphics[width=9.0cm]{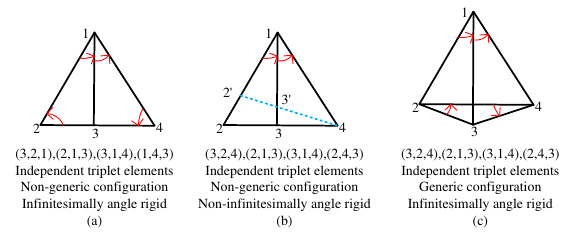}
	\caption{Relationship among generic configuration, independent triplet elements, and infinitesimal angle rigidity.}
		\label{AR/fig-relation}
\end{figure}

For infinitesimally angle rigid angularities, we now discuss when its number of angles in $\mathcal A$ becomes the minimum. Towards this end, we need to clarify what we mean by minimal angle rigidity.

\begin{definition}\label{AR/def-minimal}
An angularity $\mathbb{A}(\mathcal{V},\mathcal{A},p)$ is \emph{minimally angle rigid} if it is angle rigid and fails to remain so after removing any element in $\mathcal{A}$.
\end{definition}

\begin{definition}\label{AR/def-minimalrigidity}
	An angularity $\mathbb{A}(\mathcal{V},\mathcal{A},p)$ is \emph{infinitesimally minimally angle rigid} if it is infinitesimally angle rigid and minimally angle rigid.
\end{definition}

Since  $\text{Rank}[B(p)]\leq 2N-4$, the minimum number of angle constraints in $f_{\mathcal{A}}(p)$ to maintain infinitesimal angle rigidity is exactly $2N-4$. So we immediately  have the following lemma.

\begin{lemma}\label{AR/lem-minimal}
	An angularity $\mathbb{A}(\mathcal{V},\mathcal{A},p)$ is infinitesimally minimally angle rigid if and only if it is infinitesimally angle rigid and $|\mathcal{A}|=2N-4$.
\end{lemma}

For an infinitesimally minimally distance rigid framework, there must exist a vertex associated with fewer than 4 distance constraints \cite{tay1985generating,whiteley1996some}; otherwise, the total number of distance constraints will be at least $2N$ and thus greater than the minimum number $2N-3$. This property is critical for the success of the Henneberg construction method in order to generate an arbitrary infinitesimally minimally distance rigid framework\cite{laman1970graphs,tay1985generating}. However, for an infinitesimally minimaly angle rigid angularity, the situation is more challenging, which in fact prevents drawing similar conclusions as the Henneberg construction does for distance rigidity. To be more precise, we have the following lemma.

\begin{lemma}\label{AR/lem-minimalvertex}
For an infinitesimally minimally angle rigid angularity $\mathbb{A}(\mathcal{V},\mathcal{A},p)$  with $|\mathcal{A}|= 2N-4$, it must have a vertex involved in more than one but fewer than 6 angle constraints.
\end{lemma}

\begin{proof}
	If every vertex is involved in at least 6 angle constraints, then the total number of angle constraints is at least $|\mathcal{A}|\geq \frac{6N}{3}=2N$, which contradicts Lemma \ref{AR/lem-minimal}. Then for that vertex, which has fewer than 6 angle constraints, if it is involved in only one angle constraint, then it is flexible with respect to the rest of the angularity, which contradicts the property of angle rigid. So there must be at least one vertex that is involved in 2, 3, 4 or 5 angle constraints.
\end{proof}

In the following example, we show an infinitesimally minimally angle rigid angularity, whose vertices are all involved in 5 angle constraints Fig. \ref{AR/fig-showdegree5}.

\begin{figure}[H]
	\centering
	\includegraphics[width=6.0cm]{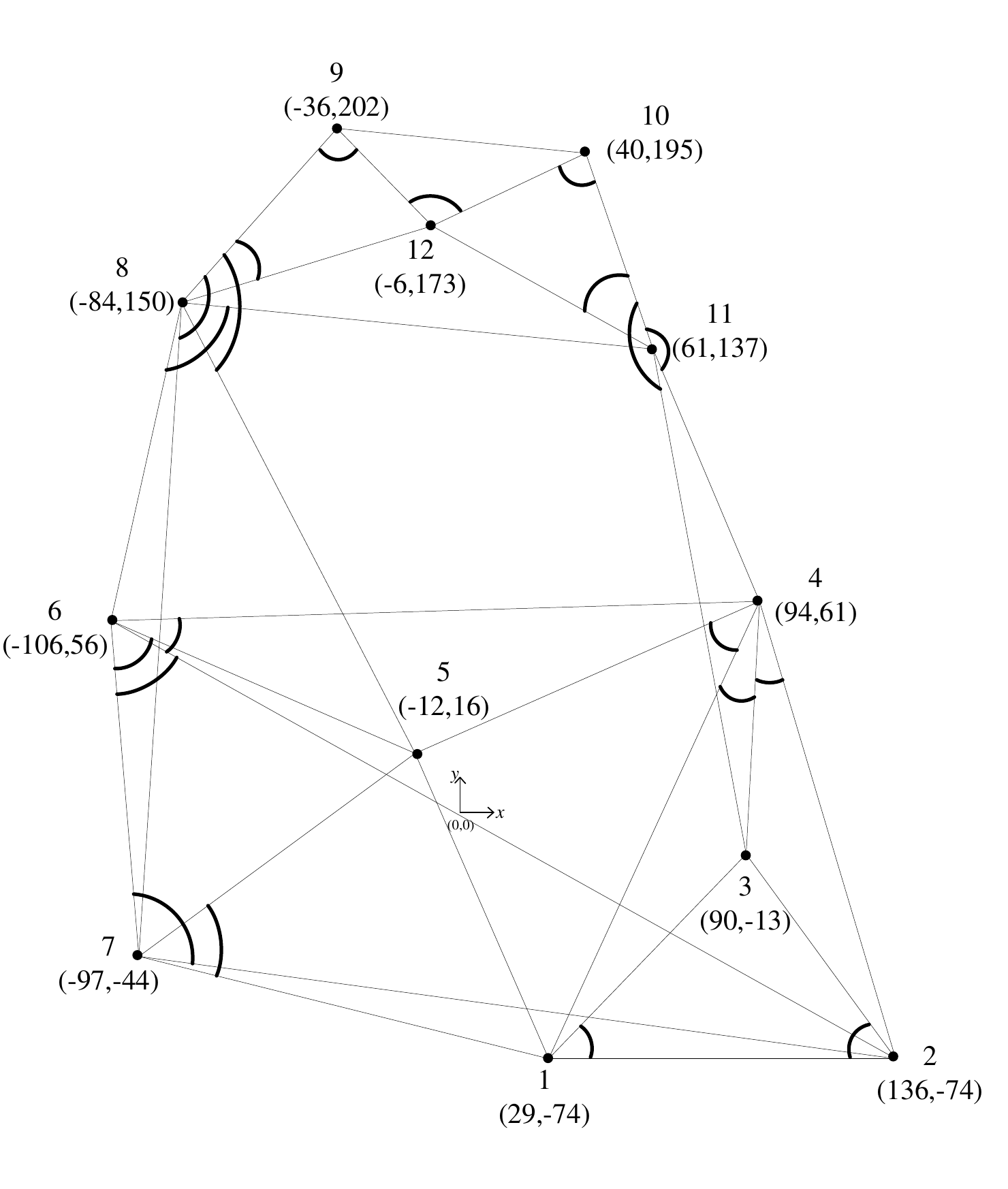}
	\caption{All vertices are involved in 5 angle constraints.}
	\label{AR/fig-showdegree5}
	\centering
\end{figure}

Note that if an angularity $\mathbb{A}(\mathcal{V},\mathcal{A},p)$ with a generic $p$ is infinitesimally minimally angle rigid, then $|\mathcal{A}|=2N-4$, and more importantly, the angles in $\mathcal A$ need to be independent; this also implies that those situations listed after Theorem \ref{AR/theo-infinitesimalrigidity}, namely cyclic angles, angles around a vertex, and overly constrained subsets, cannot show up. In the following section, we show how to apply the angle rigidity theory that we have developed for multi-agent formation control.

\section{Application in multi-agent planar formations}

To  achieve a planar formation by a group of mobile robots, many formation control algorithms have been reported, most of which require the measurement of relative positions\cite{lin2005necessary,anderson2017formation,jing2019angle} or aligned bearings\cite{zhao2016bearing,trinh2019bearing}. In this section we demonstrate how to stabilize a multi-agent planar formation  using only angle measurements with the help of the angle rigidity theory that we have just developed.

For an agent $i$ moving in the plane, we consider its dynamics are governed by
\begin{equation}
\dot{p}_{i}=\begin{bmatrix} \dot x_{i} \\ \dot y_{i} \end{bmatrix}=u_i,i=1,\cdots, N, \label{AR/eq-eq6}
\end{equation}
where $p_i=[x_i,y_i]^{T}\in \R^2$ denotes agent $i$'s position, and $u_i$ is the control input to be designed. Agent $i$ can only measure angles; to be more specific, with respect to another agent $j$, it can only measure the angle $\phi_{ij}\in [0,2\pi)$  with respect to another agent $j$ evaluated counter-clockwise from the   $x$-axis of its own local coordinate system of choice that is fixed to the ground. 

To introduce the control law, we define the bearing $z_{ij}\in \R^2$ to be the unit vector pointing from agent $i$ to $j$ represented in agent $i$'s local coordinate system, i.e., \begin{equation}
z_{ij}=\frac{p_j-p_i}{\|p_j-p_i\|}=\begin{bmatrix} \cos \phi_{ij} \\ \sin \phi_{ij} \end{bmatrix}.
\end{equation}
In the triangle $\triangle ijk$ shown below in Fig. \ref{AR/fig-anglemeasure}, the interior angle $\alpha_{i}$ can be computed by
\begin{equation}
\alpha_{i}=\measuredangle kij=\arccos(z_{ij}^Tz_{ik}),
\end{equation}
using  bearings $z_{ij}$ and $z_{ik}$. Note that the $x$-axes of agents $i$, $j$ and $k$ do not need to align. 

\begin{figure}[H]
	\centering\includegraphics[width=1.7in]{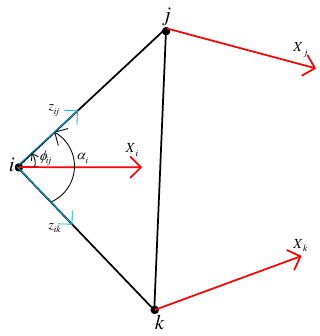} 
	\caption{The angle measurements.}
		\label{AR/fig-anglemeasure}
\end{figure}

We construct the desired planar formation through a sequence of Type-I vertex additions (Case 3) from a generically angle rigid 3-vertex angularity, which is globally angle rigid according to Proposition \ref{AR/prop-globalrigid}. In other words, in an $N$-agent formation, we label the agents by 1 to $N$. Then agents 1, 2, 3 aim at forming the first triangular shape, and each of agents 4 to $N$ aims at achieving two desired angles formed with other three agents, see Fig. \ref{AR/fig-formationproble}. By repeatedly adding new agents through the Type-I vertex addition operation, the aim is to achieve the desired angle rigid formation specified as follows. For agents 1 to 3
\begin{equation}
\lim\nolimits_{t \to \infty}e_1(t)=\lim\nolimits_{t \to \infty}(\alpha_{312}(t)-\alpha_{312}^{*})= 0,
\end{equation}
\begin{equation}
\lim\nolimits_{t \to \infty}e_2(t)=\lim\nolimits_{t \to \infty}(\alpha_{123}(t)-\alpha_{123}^{*})= 0,
\end{equation}
\begin{equation}
\lim\nolimits_{t \to \infty}e_3(t)=\lim\nolimits_{t \to \infty}(\alpha_{231}(t)-\alpha_{231}^{*})= 0,
\end{equation}
where $\alpha_{jik}^{*}\in (0,\pi),i,j,k\in\{1,2,3\}$ denote agent $i$'s desired angle formed with agents $j,k$. For agents 4 to $N$
\begin{equation}
\lim\nolimits_{t \to \infty}e_{i1}(t)=\lim\nolimits_{t \to \infty}(\alpha_{j_1ij_2}(t)-\alpha_{j_1ij_2}^{*})= 0,
\end{equation}
\begin{equation}
\lim\nolimits_{t \to \infty}e_{i2}(t)=\lim\nolimits_{t \to \infty}(\alpha_{j_2ij_3}(t)-\alpha_{j_2ij_3}^{*})= 0,
\end{equation}
where $i=4,\cdots, N$, $j_1<i,j_2<i,j_3<i$, and $\alpha_{j_1ij_2}^{*}\in (0,\pi),\alpha_{j_2ij_3}^{*}\in (0,\pi)$ denote agent $i$'s two desired angles formed with agents $j_1,j_2,j_3\in \{1,2,...,N-1\}$.

\begin{figure}[H]
	\centering\includegraphics[width=2.0in]{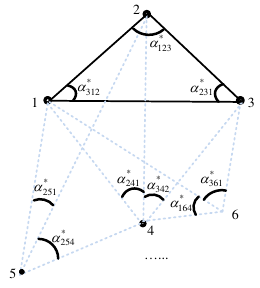} 
	\caption{Problem formulation.}
		\label{AR/fig-formationproble}
\end{figure}


\subsection{Triangular formation control for agents 1 to 3}
To achieve the desired angles for agents 1 to 3, we design their formation control laws
\begin{align}\label{maneuver-4}
u_i=&-{(\alpha_{i}-\alpha_{i}^{*})(z_{i(i+1)}+z_{i(i-1)})},
\end{align}
where $i\in \{1,2,3\}$, $z_{i(i+1)}=z_{31}$ when $i=3$ and  $z_{i(i-1)}=z_{13}$ when $i=1$, and $\alpha_{i}$ represents $\alpha_{(i-1)i(i+1)}$ for conciseness.

	To obtain the convergence of the relative angle errors, we first analyze the dynamics of the relative angle errors $e_i(t),i=1,2,3$.
	Different from \cite{basiri2010distributed}, we use the dot product of two bearings to obtain the angle dynamics. Take agent 1 as an example,
	\begin{align}
	\frac{\text{d}(\cos \alpha_1)}{\text{d}t}&=-\sin(\alpha_1)\dot{\alpha}_1=	\frac{\text{d}(z_{12}^Tz_{13})}{\text{d}t}\notag \\
	&=(\dot z_{12})^Tz_{13}+(z_{12})^T\dot z_{13}.
	\end{align}
	Considering that for $x\in \R^2,x\neq 0, \frac{\text{d}}{\text{d}t}(\frac{x}{\|x\|})=\frac{P_{x/\|x\|}}{\|x\|} \dot x$ where $P_{x/\|x\|}=I_2-\frac{x}{\|x\|}\frac{x^T}{\|x\|}$, one has
	\begin{equation}
	\dot z_{12}=\frac{P_{z_{12}}}{l_{12}}(\dot p_2-\dot p_1).
	\end{equation}
In view of (\ref{maneuver-4}), it follows
	\begin{align}
	\dot z_{12}=&\frac{P_{z_{12}}}{l_{12}}(u_2-u_1) \\ \notag
	=&\frac{P_{z_{12}}}{l_{12}}[-(\alpha_{2}-\alpha_{2}^{*})(z_{23}+z_{21}) +(\alpha_{1}-\alpha_{1}^{*})(z_{13}+z_{12})].
	\end{align}
So
	\begin{align}
	& (\dot z_{12})^Tz_{13} \\ 
	=&[(\alpha_{1}-\alpha_{1}^{*})(z_{13}+z_{12})-(\alpha_{2}-\alpha_{2}^{*})(z_{23}+z_{21})]^T\frac{P_{z_{12}}}{l_{12}}z_{13} \notag \\
	=&\frac{\sin ^2(\alpha_{1})(\alpha_{1}-\alpha_{1}^{*})\notag -(\cos\alpha_{3}+\cos\alpha_{1}\cos\alpha_{2})(\alpha_{2}-\alpha_{2}^{*})}{l_{12}}. \notag 
	\end{align}
Since
	\begin{align}
	&{\cos\alpha_{3}+\cos\alpha_{1}\cos\alpha_{2}}={-\cos(\alpha_{1}+\alpha_{2})+\cos\alpha_{1}\cos\alpha_{2}}\notag \\
	=&\sin\alpha_{2}{\sin\alpha_{1}},
	\end{align}
it follows 
	\begin{align}
	(\dot z_{12})^Tz_{13} =\frac{\sin \alpha_{1}}{l_{12}}[(\alpha_{1}-\alpha_{1}^{*})\sin (\alpha_{1})-(\alpha_{2}-\alpha_{2}^{*})\sin \alpha_{2}]. \notag
	\end{align}
	Similarly, one  gets
	\begin{align}
	&(z_{12})^T\dot z_{13} \notag \\
	=&(z_{12})^T\frac{P_{z_{13}}}{l_{13}}(u_3-u_1)\notag \\
	=&\frac{\sin \alpha_{1}}{l_{13}}[(\alpha_{1}-\alpha_{1}^{*})\sin \alpha_{1} -(\alpha_{3}-\alpha_{3}^{*})\sin \alpha_{3}]. 
	\end{align}
So agent 1's closed-loop angle dynamics are
	\begin{align} \label{taf-1}
	\dot \alpha_{1}=&-	\frac{1}{{\sin\alpha_{1}}}\frac{\text{d}(\cos \alpha_1)}{\text{d}t}=-\frac{(\dot z_{12})^Tz_{13}+(z_{12})^T\dot z_{13}}{\sin\alpha_{1}}\notag  \\
	=&-\sin (\alpha_{1})(\frac{1}{l_{12}}+\frac{1}{l_{13}})(\alpha_{1}-\alpha_{1}^{*})\notag \\
	&+\frac{\sin \alpha_{2}}{l_{12}}(\alpha_{2}-\alpha_{2}^{*})+\frac{\sin \alpha_{3}}{l_{13}}(\alpha_{3}-\alpha_{3}^{*}).
	\end{align}
Similarly,
	\begin{align}
	\dot \alpha_{2}=&-\sin (\alpha_{2})(\frac{1}{l_{21}}+\frac{1}{l_{23}})(\alpha_{2}-\alpha_{2}^{*})\notag \\
	&+\frac{\sin \alpha_{1}}{l_{21}}(\alpha_{1}-\alpha_{1}^*)
	+\frac{\sin \alpha_{3}}{l_{23}}(\alpha_{3}-\alpha_{3}^*),
	\end{align}
	\begin{align} \label{maneuver-23}
	\dot \alpha_{3}=&-\sin (\alpha_{3})(\frac{1}{l_{31}}+\frac{1}{l_{32}})(\alpha_{3}-\alpha_{3}^{*})\notag \\
	&+\frac{\sin \alpha_{1}}{l_{31}}(\alpha_{1}-\alpha_{1}^*)+\frac{\sin \alpha_{2}}{l_{32}}(\alpha_{2}-\alpha_{2}^*).
	\end{align}
	
Writing (\ref{taf-1})-(\ref{maneuver-23}) into a compact form, one has the following closed-loop triangular formation dynamics
	\begin{align}\label{maneuver-9}
	\dot e_f&=[
	\dot \alpha_1 ~~ \dot \alpha_2 ~~ \dot \alpha_3]^T
	=F(e_f)e_f \notag \\
	&=\begin{bmatrix}
	-g_1 & f_{12} & f_{13}\\
	f_{21} & -g_2 & f_{23} \\
	f_{31} & f_{32} & -g_{3}
	\end{bmatrix}\begin{bmatrix}
	\alpha_1 - \alpha_1^* \\  \alpha_2- \alpha_2^*  \\ \alpha_3- \alpha_3^*  \end{bmatrix},
	\end{align} 
	where 
	\begin{equation}
	e_f=\begin{bmatrix}
	\alpha_1 - \alpha_1^* &  \alpha_2- \alpha_2^*  & \alpha_3- \alpha_3^*
	\end{bmatrix}^T,\notag
	\end{equation}
	\begin{equation}
	g_{i}=\sin (\alpha_{i})({1}/{l_{i(i+1)}}+{1}/{l_{i(i-1)}}), \notag
	\end{equation}
	\begin{equation}
	f_{ij}=\sin (\alpha_{j})/{l_{ij}}. \notag
	\end{equation}

	To guarantee that the triangular formation system under the control law (\ref{maneuver-4}) is well defined, we first prove that  no collinearity and collision will take place under (\ref{maneuver-9}) if the formation is not collinear initially.
	\begin{lemma}\label{AR/lemma-nocollinearity}
		For the three-agent formation, if the initial formation is not collinear, it will not become collinear  for $t>0$ under the angle dynamics (\ref{maneuver-9}).
	\end{lemma}
	\begin{proof}
		Consider the  manifold $\mathcal{M}_a=\{(\alpha_1,\alpha_2,\alpha_3)| {\alpha_1+\alpha_2+\alpha_3=\pi}, {0<\alpha_1<\pi} , {0<\alpha_2<\pi}, \\ \text{and} ~0<\alpha_3<\pi \}$
		which is an open set.
		To show $\mathcal{M}_a$ is positively invariant, we show that for any $\alpha_i \in \mathcal{M}_a,i=1,2,3$, it is impossible for $\alpha_i$ to escape $\mathcal{M}_a$. Consider the boundary states $\alpha_i(t)=\pi-\varepsilon_1$ with $\varepsilon_1=0^+$, $\alpha_{i+1}(t)=\varepsilon_2=0^+$, $\alpha_{i-1}(t)=\varepsilon_3=0^+$, $\varepsilon_1=\varepsilon_2+\varepsilon_3$.
		
		According to (\ref{maneuver-9}), one has 
		\begin{equation}
		\dot e_i=-g_ie_i+f_{i(i+1)}e_{i+1}+f_{i(i-1)}e_{i-1}.
		\end{equation}
		
		Since $0<\alpha_i^*<\pi$ and $\alpha_i^*$ is bounded away from 0 and $\pi$, one has
		\begin{equation}
		g_ie_i=g_i(\alpha_i-\alpha_i^*)>0,
		\end{equation}
		\begin{equation}
		f_{i(i+1)}e_{i+1}=f_{i(i+1)}(\alpha_{i+1}-\alpha_{i+1}^*)< 0,
		\end{equation}
		\begin{equation}
		f_{i(i-1)}e_{i-1}=f_{i(i-1)}(\alpha_{i-1}-\alpha_{i-1}^*)< 0,
		\end{equation}
		which implies that $\dot e_i(t)<0$. Thus when $\alpha_i(t)$ is close to $\pi$,  $\alpha_i(t)$ will decrease, which implies that $\mathcal{M}_a$ is positively invariant. 
	\end{proof}

	\begin{lemma}\label{AR/lemma-nocollision}
	For the three-agent formation, if the initial angles $\alpha_{i}\neq 0,i=1,2,3$, no collision will take place for $t>0$ under the formation control law (\ref{maneuver-4}).
\end{lemma}
\begin{proof}
	Suppose on the contrary that collision may happen between agents $i$ and $j$ at $t=t_1$. Then one of the following two cases will take place.
\begin{figure}[H]
	\centering\includegraphics[width=2in]{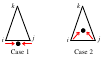} 
	\caption{Collision cases.}
	\label{AR/fig-formationcollision}
\end{figure}	
	
	For the first case, $\dot p_i(t_1)=-\gamma\dot p_j(t_1)$ where $\gamma$ is a positive constant. Note that the moving direction of agent $i$ under the control law (\ref{maneuver-4}) is always the bisector of the interior angle $\alpha_{i}$.  According to Lemma \ref{AR/lemma-nocollinearity}, no collinearity will happen for $t>0$ which implies that $z_{ik}(t)\neq -z_{jk}(t)$ for $t>0$. According to the control law (\ref{maneuver-4}), $\dot p_i(t_1)=-\gamma\dot p_j(t_1)$ requires $z_{ik}(t_1)= -z_{jk}(t_1)$ which is impossible for $t>0$.
	
	For the second case, since agents $i$ and $j$ move towards the inside of the triangle, it follows from  the control law (\ref{maneuver-4}) that $\frac{\pi}{2}-\varepsilon_1=\alpha_{i}(t_1^{-})<\alpha_{i}^*$ and $\frac{\pi}{2}-\varepsilon_2=\alpha_{j}(t_1^{-})<\alpha_{j}^*$, where $\varepsilon_1=0^+$ and $\varepsilon_2=0^+$. Then, $\alpha_{i}^*+\alpha_{j}^*+\alpha_{k}^*=\pi>\pi+\alpha_k^*-\varepsilon_1-\varepsilon_2$, which contradicts the fact that $\alpha_{k}^*$ is bounded away from 0.
\end{proof}
     	
%

	Now, we give the main result for the convergence of the triangular formation.
\begin{theorem}
	For the triangular  formation under the control law (\ref{maneuver-4}),  if $\alpha_{i}(0)\neq 0$ and the initial angle errors $e_i(0),i=1,2,3$ are sufficiently small,
	the angle errors $e_i$ and agents' control input $u_i(t)$   converge exponentially to zero.
\end{theorem}

\begin{proof}
From Lemmas \ref{AR/lemma-nocollinearity} and \ref{AR/lemma-nocollision}, no collinearity and collision will take place since $ \sin(\alpha_{i})\neq 0, l_{ij}\neq 0, \forall i,j=1,2,3$, which guarantees that the closed-loop system under the control law (\ref{maneuver-4}) is well defined. Since $e_1+e_2+e_3\equiv 0$, the angle dynamics (\ref{maneuver-9}) can be reduced to
	\begin{align}\label{AR/eq-reduceddynamics}
	\dot e_s=\begin{bmatrix}
	\dot e_1 \\ \dot e_2
	\end{bmatrix}=\begin{bmatrix}
	-(g_1+f_{13}) & f_{12}-f_{13} \\
	f_{21}-f_{23} & -(g_2+f_{23})  \\
	\end{bmatrix}\begin{bmatrix}
	e_1 \\  e_2  \end{bmatrix}=F_s(e_s)e_s.
	\end{align}

	
Let $\mathbb{U}\in \R^2$ denote  a neighborhood of the origin $\{e_1=e_2=0\}$, in which we investigate the local stability of (\ref{AR/eq-reduceddynamics}). Linearizing (\ref{AR/eq-reduceddynamics}) around the origin, we obtain
	\begin{equation}
	\dot e_s=L_1(\alpha^*)e_s,
	\end{equation}
	where  $L_1(\alpha^*)=F_s(e_s)|_{e_s=0}$. Then, one has
	\begin{equation}
	\text{tr}(L_1(\alpha^*))=-g_1-f_{13}-g_2-f_{23}<0, \label{AR/eq-tr}
	\end{equation}
	\begin{align}
	\text{det}(L_1(\alpha^*))=&(g_1+f_{13})(g_2+f_{23})-(f_{21}-f_{23})(f_{12}-f_{13})\notag \\
	>&g_1f_{23}+g_2f_{13}+f_{21}f_{13}+f_{12}f_{23}>0, \label{AR/eq-det}
	\end{align}
	where we have used the fact that $g_1g_2>f_{21}f_{12}$, and $\text{tr()}$ and $\text{det()}$ denote the trace and determinant of a square matrix, respectively. According to (\ref{AR/eq-tr}) and (\ref{AR/eq-det}), one has that $L(\alpha^*)$ is Hurwitz. 
	According to the Lyapunov Theorem\cite[Theorem 4.6]{khalil2002nonlinear}, there always exists positive definite matrices $P_1\in \R^{2\times 2}$ and $Q_1\in \R^{2\times 2}$ such that $-Q_1=P_1L_1(\alpha^*)+L^T_1(\alpha^*)P_1$. Design the Lyapunov function candidate as
	\begin{equation}
	V_1=e_s^TP_1e_s.
	\end{equation}
	Taking the time-derivative of $V_1$ yields
	\begin{equation}
	\dot V_1=-e_s^TQ_1e_s\leq -\frac{\lambda_{\min} (Q_1)}{\lambda_{\max}(P_1)}V_1.
	\end{equation}
	Then, one has
	\begin{equation}
	e_1^2+e_2^2=\|e_s\|^2\leq \frac{V_1}{\lambda_{\min}(P_1)}\leq \frac{V _1(0)}{\lambda_{\min}(P_1)}e^{-\frac{\lambda_{\min} (Q_1)}{\lambda_{\max}(P_1)}t}. \label{AR/agent1convergence}
	\end{equation}
	Also, one has
	\begin{equation}
	e_3^2=e_1^2+e_2^2+2e_1e_2\leq 2(e_1^2+e_2^2)\leq \frac{2V _1(0)}{\lambda_{\min}(P_1)}e^{-\frac{\lambda_{\min} (Q_1)}{\lambda_{\max}(P_1)}t},
	\end{equation}
	which implies that $e_i$ under the dynamics (\ref{maneuver-9}) is exponentially stable when the initial states lie in $\mathbb{U}$. According to (\ref{maneuver-4}), $\|u_i\|\leq 2|e_i|$ also converge to zero at an exponential rate.
\end{proof}

After proving the first three agents converge to the desired formation, we now look at the remaining agents.

{ 
\subsection{Adding agents 4 to N in sequence}

In this subsection, we consider that agent $i, i=4,...,N$, are added to the formation through  the Type-I vertex addition operation with two desired angles $\measuredangle j_1ij_2$ and $\measuredangle j_2ij_3$, $j_1<i$, $j_2<i$, and $j_3<i$. 
For agents $i=4,...,N$, the control algorithm is designed to be
\begin{align}
u_i=&-{(\alpha_{j_1ij_2}-\alpha_{j_1ij_2}^{*})(z_{ij_1}+z_{ij_2})}\notag \\
&- {(\alpha_{j_2ij_3}-\alpha_{j_2ij_3}^{*})(z_{ij_2}+z_{ij_3})}, \label{AR/eq-agenti}
\end{align}
where $\alpha_{j_1ij_2}^{*}\in (0,\pi)$ and $\alpha_{j_2ij_3}^{*}\in (0,\pi)$, $j_1<i,j_2<i,j_3<i$ are the two desired angles.

Now, we present the main result.
\begin{theorem}\label{AR/Theo-agents4N}
	Consider a formation of  $N>3$ agents, each of which is governed by (\ref{AR/eq-eq6}). Suppose $\dot p_1,\dot p_2,\dot p_3$ are sufficiently small and the sub-formation  of  1, 2, 3 converges to the desired triangular shape exponentially fast. For   agent $i,4\leq i\leq N$,  if  the initial distances $l_{ij_1}(0)$, $l_{ij_2}(0)$, $l_{ij_3}(0)$ are sufficiently bounded away from zero, the initial angle errors $e_{i1}(0)$ and $e_{i2}(0)$
	are sufficiently small  and $l_{ij_1}^*>l_{ij_2}^*,l_{ij_3}^*>l_{ij_2}^*$, then under (\ref{AR/eq-agenti}), the formation achieves its desired shape exponentially fast.
\end{theorem}

To prove this theorem, we use induction. Towards this end, we need to 
first prove that the 4-agent formation of  1 to 4 converges to the desired shape exponentially fast. For the 4-agent formation, the control algorithm (\ref{AR/eq-agenti}) can be written as  
\begin{equation}
u_4=-{(\alpha_{241}-\alpha_{241}^{*})(z_{41}+z_{42})}- {(\alpha_{342}-\alpha_{342}^{*})(z_{42}+z_{43})}. \label{icca-6}
\end{equation}

\begin{lemma}\label{AR/Lemma-agent4}
Suppose $\dot p_1,\dot p_2,\dot p_3$ are sufficiently small and the sub-formation  of  1, 2, 3 converges to the desired triangular shape exponentially fast. Under the control algorithm (\ref{icca-6}) for  agent 4,  if  the initial distances $l_{4i}(0)$ are sufficiently bounded away from zero, the initial angle errors $e_{41}(0)$ and $e_{42}(0)$ 
	are sufficiently small  and $l_{41}^*>l_{42}^*,l_{43}^*>l_{42}^*$, then $e_{41}(t)$ and $e_{42}(t)$ converges to zero exponentially fast.
\end{lemma}
 
	\begin{proof}
	To analyze the stability of the relative angle errors $e_{41}$ and $e_{42}$ under the control algorithm (\ref{icca-6}), we first calculate the error dynamics of $e_{41}$ and $e_{42}$. Since
	\begin{align}
\frac{\text{d}(\cos \alpha_{241})}{\text{d}t}	&=-\sin(\alpha_{241})\dot{\alpha}_{241}=\frac{\text{d}(z_{41}^Tz_{42})}{\text{d}t}\notag \\
	&=(\dot z_{41})^Tz_{42}+(z_{41})^T\dot z_{42} ,\label{icca-9}
	\end{align}
and similarly
	\begin{align}
	\dot z_{41}=\frac{P_{z_{41}}}{l_{41}}(\dot p_1-\dot p_4)=\frac{P_{z_{41}}}{l_{41}}u_1-\frac{P_{z_{41}}}{l_{41}}u_4,
	\end{align}
	we have
	\begin{align}
	&(\dot z_{41})^Tz_{42}\notag \\
	=&u_1^T\frac{P_{z_{41}}}{l_{41}}z_{42}-\frac{u_4^T}{l_{41}} (I_2-z_{41}z_{41}^T)z_{42}\notag \\
	=&u_1^T\frac{P_{z_{41}}}{l_{41}}z_{42}+\frac{u_4^Tz_{41}\cos \alpha_{241}-u_4^Tz_{42}}{l_{41}}\notag \\
	=&u_1^T\frac{P_{z_{41}}}{l_{41}}z_{42}-\frac{[(\alpha_{241}-\alpha_{241}^{*})(\cos \alpha_{241}+\cos^2 \alpha_{241})}{l_{41}}] \notag \\
	&-\frac{[(\alpha_{342}-\alpha_{342}^{*})(\cos^2 \alpha_{241}+\cos \alpha_{241}\cos \alpha_{341})]}{l_{41}} \notag \\
	&+\frac{[(\alpha_{241}-\alpha_{241}^{*})(\cos \alpha_{241}+1)}{l_{41}}] \notag \\
	&+\frac{[(\alpha_{342}-\alpha_{342}^{*})(1+\cos \alpha_{342})]}{l_{41}} \notag \\
	=&u_1^T\frac{P_{z_{41}}}{l_{41}}z_{42}+\frac{(\alpha_{241}-\alpha_{241}^{*})\sin^2\alpha_{241}}{l_{41}}\notag \\
	&+\frac{(\alpha_{342}-\alpha_{342}^{*})(\sin^2\alpha_{241}+\sin^2\alpha_{241}\cos \alpha_{342})}{l_{41}}\notag \\
	&+\frac{(\alpha_{342}-\alpha_{342}^{*})\cos \alpha_{241}\sin \alpha_{241}\sin \alpha_{342}}{l_{41}},
	\end{align}
and
	\begin{align}
	&z_{41}^T\dot z_{42}\notag \\
	=&u_2^T\frac{P_{z_{42}}}{l_{42}}z_{41}-z_{41}^T\frac{I_2-z_{42}z_{42}^T}{l_{42}}u_{4}\notag \\
	=&u_2^T\frac{P_{z_{42}}}{l_{42}}z_{41}+\frac{(\alpha_{241}-\alpha_{241}^{*})\sin^2\alpha_{241}}{l_{42}}\notag \\
	&+\frac{(\alpha_{342}-\alpha_{342}^{*})(-\sin \alpha_{241}\sin \alpha_{342}) }{l_{42}}.
	\end{align}

	Then from  (\ref{icca-9}), it follows
	\begin{align}
	\dot \alpha_{241}=&-\frac{1}{\sin\alpha_{241}}\frac{\text{d}(\cos \alpha_{241})}{\text{d}t}=-\frac{\dot z_{41}^Tz_{42}+z_{41}^T\dot z_{42}}{\sin\alpha_{241}} \notag \\
	=&-\sin (\alpha_{241})(\frac{1}{l_{41}}+\frac{1}{l_{42}})(\alpha_{241}-\alpha_{241}^{*})\notag \\
	&-\frac{(\alpha_{342}-\alpha_{342}^{*})(\sin \alpha_{241}+\sin \alpha_{341}) }{l_{41}}+\frac{u_1^TP_{z_{41}}z_{42}}{l_{41}}\notag \\
	&+\frac{(\alpha_{342}-\alpha_{342}^{*}) \sin\alpha_{342}}{l_{42}}+\frac{u_2^TP_{z_{42}}z_{41}}{l_{42}}.\label{icca-37}
	\end{align}

	Analogously, 
	\begin{align} \label{icca-36}
	\dot \alpha_{342}=&-\frac{1}{\sin\alpha_{342}}\frac{\text{d}(\cos \alpha_{342})}{\text{d}t}=-\frac{\dot z_{42}^Tz_{43}+z_{42}^T\dot z_{43}}{\sin\alpha_{342}} \\
	=&-\sin(\alpha_{342})(\frac{1}{l_{43}}+\frac{1}{l_{42}})(\alpha_{342}-\alpha_{342}^{*})\notag \\
	&-\frac{(\alpha_{241}-\alpha_{241}^{*})(\sin \alpha_{342}+\sin \alpha_{341}) }{l_{43}}\notag \\
	&+\frac{(\alpha_{241}-\alpha_{241}^{*})\sin\alpha_{241} }{l_{42}}+u_2^T\frac{P_{z_{42}}}{l_{42}}z_{43}+u_3^T\frac{P_{z_{43}}}{l_{43}}z_{42}. \notag
	\end{align}
	
	By combining (\ref{icca-37}) and (\ref{icca-36}), one has the compact form
	\begin{align}\label{icca-10}
	\dot e_4&=[
	\dot \alpha_{241} ~~ \dot \alpha_{342} ]^T\notag \\
	&=F_4(e_4)
	e_{4}
	+W(e_4)
	U(u_1,u_2,u_3)
	\notag \\
	&=\begin{bmatrix}
	j_{11} & j_{12} \\
	j_{21} & j_{22}  \\
	\end{bmatrix}\begin{bmatrix}
	e_{41}\\  e_{42}  \end{bmatrix}+\begin{bmatrix}
	w_{11} & w_{12} & w_{13} \\
	w_{21} &w_{22} & w_{23}
	\end{bmatrix}\begin{bmatrix}
	u_1 \\ u_2 \\u_3
	\end{bmatrix},
	\end{align} 
	where 
	$j_{11}=-\frac{\sin\alpha_{241}}{l_{41}}-\frac{\sin\alpha_{241}}{l_{42}}$, $j_{22}=-\frac{\sin\alpha_{342}}{l_{43}}-\frac{\sin\alpha_{342}}{l_{42}}$,
	$j_{12}=-\frac{\sin(\alpha_{241})+\sin (\alpha_{341}) }{l_{41}}+\frac{\sin\alpha_{342} }{l_{42}}$,
	$j_{21}=-\frac{\sin (\alpha_{342})+\sin( \alpha_{341}) }{l_{43}}+\frac{\sin\alpha_{241}}{l_{42}}$,
	$w_{11}=z_{42}^T\frac{P_{z_{41}}}{l_{41}}$, $w_{12}=z_{41}^T\frac{P_{z_{42}}}{l_{42}}$, $w_{13}=0$, $w_{21}=0$, $w_{22}=z_{43}^T\frac{P_{z_{42}}}{l_{42}}$,  $w_{23}=z^T_{42}\frac{P_{z_{43}}}{l_{43}}$, and $U(u_1,u_2,u_3)=[u_1^T, u_2^T, u_3^T]^T$.

	For $U(u_1,u_2,u_3)$, one has 
	\begin{align} \label{AR/agentsinput123}
	\|U(u_1,u_2,u_3)\|^2&=\|u_1\|^2+\|u_2\|^2+\|u_3\|^2\notag \\
	&\leq 4(e_1^2+e_2^2+e_3^2)\notag \\
	&\leq \frac{12V _1(0)}{\lambda_{\min}(P_1)}e^{-\frac{\lambda_{\min} (Q_1)}{\lambda_{\max}(P_1)}t},
	\end{align}
	which implies that $U(u_1,u_2,u_3)$ exponentially converges to zero. Since $e_i(0),i=1,2,3$ is sufficiently small, $V_1(0)$ is sufficiently small. Therefore, $\|U(u_1,u_2,u_3)\|$  is always  sufficiently small  and there exists a finite time $T$ such that 
	\begin{equation}
	\|U(u_1(T),u_2(T),u_3(T))\|^2\leq \frac{12V _1(0)}{\lambda_{\min}(P_1)}e^{-\frac{\lambda_{\min} (Q_1)}{\lambda_{\max}(P_1)}T}=\varepsilon_4, \notag
	\end{equation}
	where $\varepsilon_4=0^{+}$.
	
	When $\|W(e_4)\|$ is bounded and $\|U(u_1,u_2,u_3)\|$ is sufficiently small and exponentially converges to zero, one can first consider the stability of the following system
	\begin{equation}
	\dot e_4=F_4(e_4)e_4. \label{AR/eq-nonminal}
	\end{equation}
Since the initial angle errors $e_{41}(0)$ and $e_{42}(0)$  are sufficiently small, it can be easily verified that in a small neighborhood of the origin $\{e_{41}=0,e_{42}=0\}$, (\ref{AR/eq-nonminal}) can be linearized by 
	\begin{equation}
	\dot e_{4}=L_{2}(\alpha^*)e_{4},
	\end{equation}
	where 
	$L_2(\alpha^*)=F_4(e_4)|_{e_4=0}$. Then, one has
	\begin{equation}
	\text{tr}(L_2(\alpha^*))=j_{11}(\alpha^*)+j_{22}(\alpha^*)<0,
	\end{equation}
	\begin{align}
	&\text{det}(L_2(\alpha^*))\\
	=&j_{11}j_{22}-j_{12}j_{21} \notag\\
	=&\frac{l_{41}^*(\sin \alpha_{241} ^* \sin \alpha_{342}^*+\sin^2 \alpha_{342}^*+\sin \alpha_{342}^*\sin\alpha_{341}^*)}{l_{41}^*l_{42}^*l_{43}^*}\notag \\
	&+\frac{l_{43}^*(\sin \alpha_{241} ^* \sin \alpha_{342}^*+\sin^2 \alpha_{241}^*+\sin \alpha_{241}^*\sin\alpha_{341}^*)}{l_{42}^*l_{41}^*l_{43}^*}\notag \\
	&-\frac{l_{42}^*(\sin \alpha_{241}^*\sin\alpha_{341}^*+\sin \alpha_{341}^*\sin\alpha_{342}^*+\sin^2 \alpha_{341}^*)}{l_{41}^*l_{42}^*l_{43}^*}. \notag
	\end{align}
	
	Then, if $\text{det}(L_2(\alpha^*))>0$, one has that   $L_2(\alpha^*)$ is Hurwitz. One can check that $\text{det}(F(\alpha^*))>0$ if $l_{41}^*>l_{42}^*$ and $l_{43}^*>l_{42}^*$ hold because 
	\begin{equation}
	l_{43}^*\sin \alpha_{241}^*\sin \alpha_{341}^*>l_{42}^*\sin \alpha_{241}^*\sin \alpha_{341}^*,
	\end{equation}
	\begin{equation}
	l_{41}^*\sin \alpha_{341}^*\sin \alpha_{342}^*>l_{42}^*\sin \alpha_{341}^*\sin \alpha_{342}^*,
	\end{equation}
	and 
	\begin{align}
	\sin^2 \alpha_{341}^*=&[\sin \alpha_{241}^*\cos \alpha_{342}^*+\cos \alpha_{241}^*\sin \alpha_{342}^*]^2\notag \\
	=&\sin^2 \alpha_{241}^*\cos^2 \alpha_{342}^*+\cos^2 \alpha_{241}^*\sin^2 \alpha_{342}^*\notag \\
	&+2\sin \alpha_{241}^*\cos \alpha_{342}^*\cos \alpha_{241}^*\sin \alpha_{342}^*,
	\end{align}
	and 
	\begin{equation}
	l_{41}^*\sin^2 \alpha_{342}^*>l_{42}^*\sin^2 \alpha_{342}^*\cos^2 \alpha_{241}^*,
	\end{equation}
	\begin{equation}
	l_{43}^*\sin^2 \alpha_{241}^*>l_{42}^*\sin^2 \alpha_{241}^*\cos^2 \alpha_{342}^*,
	\end{equation}
	\begin{align}
	&l_{41}^*\sin \alpha_{241} ^* \sin \alpha_{342}^*+l_{43}^*\sin \alpha_{241} ^* \sin \alpha_{342}^*\notag \\
	&>2l_{42}^*\sin \alpha_{241} ^* \sin \alpha_{342}^*\notag \\
	&>2l_{42}^*\sin \alpha_{241}^*\cos \alpha_{342}^*\cos \alpha_{241}^*\sin \alpha_{342}^*.
	\end{align}
	
	When $L_2(\alpha^*)$ is Hurwitz, there always exists positive definite matrices $P_2\in \R^{2\times 2}$ and $Q_2\in \R^{2\times 2}$ such that $-Q_2=P_2L_2(\alpha^*)+L^T_2(\alpha^*)P_2$. Design the Lyapunov function candidate as 
	\begin{equation}
	V_2=e_4^TP_2e_4.
	\end{equation}
	Taking the time-derivative of $V_2$ along (\ref{icca-10}) yields
	\begin{align} \label{AR/agent4lyapunov}
	\dot V_2&=-e_4^TQ_2e_4+2e_4^TP_2W(e_4)
	U(u_1,u_2,u_3)\notag \\
	&\leq -\frac{\lambda_{\min} (Q_2)}{\lambda_{\max}(P_2)}V_2+2\|P_2\|\|U(u_1,u_2,u_3)\|\|e_4\|^2 \notag \\
	&	\leq -(\frac{\lambda_{\min} (Q_2)}{\lambda_{\max}(P_2)}-\frac{2\|P_2\|U_{\max}}{\lambda_{\max}(P_2)})V_2,
	\end{align}
	where $U_{\max}=\|U(u_1(0),u_2(0),u_3(0))\|$.
	Then, one has
	\begin{equation}
	\|e_4\|^2\leq \frac{V _2(0)}{\lambda_{\min}(P_2)}e^{-(\frac{\lambda_{\min} (Q_2)}{\lambda_{\max}(P_2)}-\frac{2\|P_2\|U_{\max}}{\lambda_{\max}(P_2)})t}. \label{AR/agent4convergence}
	\end{equation}
	Since $U_{\max}$ is sufficiently small, one has that $e_4$ converges exponentially to zero when  agent 4 stays around its desired location initially. According to (\ref{icca-6}), $u_4$ also converges exponentially to zero.	
Since  $e_i$ and $\|u_i\|, i=1,...,4$ always are   sufficiently small  and exponentially converge to zero, there always exists a finite time $T$ such that $e_i\leq \varepsilon_5(T)$ and $\|u_i\|\leq \varepsilon_6(T)$ with $\varepsilon_5(T)=0^+$ and $\varepsilon_6(T)=0^+$. 

	To guarantee that $\|W(e_4)\|$ is bounded and control law (\ref{icca-6}) is well defined, the collision between agent 4 and agents 1 to 3 should be avoided. Take agent 1 as an example, one has
\begin{align}
&\|p_4(t)-p_1(t)\|\notag \\
=&\|p_4(0)+\int_{0}^{t}u_4(s)\text{d}s-p_1(0)-\int_{0}^{t}u_1(s)\text{d}s\| \notag \\
\geq & \|p_4(0)-p_1(0)\|-\int_{0}^{t}\|u_1(s)-u_4(s)\|\text{d}s\notag \\
\geq& l_{14}(0)-2\int_{0}^{t}(|e_1(s)|+|e_{41}(s)|+|e_{42}(s)|)\text{d}s.\notag 
\end{align}
Since $l_{4i}$ are sufficiently bounded away from zero, there always exists a finite time $T$ such that in the time interval $[0,T)$ there is no collision between agent 4 and agents 1 to 3. Then, according to (\ref{AR/agent1convergence}) and (\ref{AR/agent4convergence}), one has
\begin{align}
&\|p_4(T)-p_1(T)\| \notag\\
\geq& l_{14}(0)-2\int_{0}^{T}(|e_1(s)|+|e_{41}(s)|+|e_{42}(s)|)\text{d}s\notag \\
\geq & l_{14}(0)-2[\sqrt{\frac{V _1(0)}{\lambda_{\min}(P_1)}}(1-e^{-\frac{\lambda_{\min} (Q_1)}{2\lambda_{\max}(P_1)}T}) \notag\\
&+\sqrt{ \frac{2V _2(0)}{\lambda_{\min}(P_2)}}(1-e^{-(\frac{\lambda_{\min} (Q_2)}{2\lambda_{\max}(P_2)}-\frac{\|P_2\|U_{\max}}{\lambda_{\max}(P_2)})T})].
\end{align}
where we have used the fact that $|e_{41}|+|e_{42}|\leq \sqrt{2(e_{41}^2+e_{42}^2)}$.
Since $V_1(0)$ and $V_2(0)$ are sufficiently small and $l_{14}(0)$ is sufficiently bounded away from zero, one has $\|p_4(T)-p_1(T)\|>0$. Then, we extend $T$ to $T'=T+\varepsilon_7>T$ with small positive $\varepsilon_7$. For the time period $[T,T')$, one also has that $\int_{T}^{T'}(|e_1(s)|+|e_{41}(s)|+|e_{42}(s)|)\text{d}s\ll \int_{0}^{T}(|e_1(s)|+|e_{41}(s)|+|e_{42}(s)|)\text{d}s$ is sufficiently small and $\|p_4(T')-p_1(T')\|>0$. Since $e_1(t)$, $e_{41}(t)$ and $e_{42}(t)$ converge at an exponential speed, one can extend $T'$ to infinity according to \cite[Theorem 2.1]{hale1980ordinary}. So,  $l_{41}(t)=\|p_4(t)-p_1(t)\|>0$ for $t>0$, which implies that $\|W(e_4)\|$ is bounded and (\ref{icca-10}) is well defined.
The proof for 4-agent formation is completed.

	\textbf{Proof of Theorem 8}
From Lemma \ref{AR/Lemma-agent4}, 4-agent formation achieves the desired  shape exponentially fast.

	Suppose for a $4<k<N$, the $k$-agent formation converges to the desired  shape exponentially fast. We need to prove that for $(k+1)$-agent formation, the relative angle errors $e_{(k+1)1}=\alpha_{j_1(k+1)j_2}-\alpha_{j_1(k+1)j_2}^{*}$ and  $e_{(k+1)2}=\alpha_{j_2(k+1)j_3}-\alpha_{j_2(k+1)j_3}^{*}$ converges to zero exponentially fast. Similar to the proof  from (\ref{icca-6}) to (\ref{AR/agent4lyapunov}), one has that the angle errors $e_{(k+1)1}$ and $e_{(k+1)2}$  exponentially converge to zero. Therefore, the control algorithm (\ref{AR/eq-agenti}) can locally stabilize agent $k+1$, i.e., the $(k+1)$-agent formation converge to the desired shape exponentially fast. So, from induction,  $N$-agent formation converges to the desired formation shape exponentially fast.	The proof for Theorem 	\ref{AR/Theo-agents4N} is completed.
\end{proof}

%
%
%
%
%

\begin{remark}
	Note that the control laws (\ref{maneuver-4}) and (\ref{AR/eq-agenti}) can be described by a unified form
	\begin{equation}
	u_i=-\sum\nolimits_{(j,i,k)\in \mathcal{A}}{(\alpha_{jik}-\alpha_{jik}^*)(z_{ij}+z_{ik})}. \label{AR/eq-uni1}
	\end{equation}
 Therefore, the unified control algorithm (\ref{AR/eq-uni1}) can locally stabilize the global angle rigid formation shape constructed through a sequence of Type-I vertex additions (Case 3) from a generically angle rigid 3-vertex angularity. {Because we aim at obtaining  local stability for multi-agent  formations in Section IV, it is reasonable that we only consider the range of the desired angles belonging to $(0,\pi)$.}
	\end{remark}
	
	\begin{remark}
		Although each agent's position in (\ref{AR/eq-eq6}) is described in the global coordinate system, it is not used in the control algorithm (\ref{AR/eq-uni1}). The control algorithm (\ref{AR/eq-uni1}) can be realized in each agent's local coordinate system since (\ref{AR/eq-uni1}) can be equivalently written as
		\begin{equation}
	R_iu_i^b=-\sum\nolimits_{(j,i,k)\in \mathcal{A}}{(\alpha_{jik}-\alpha_{jik}^*)R_i(z_{ij}^b+z_{ik}^b)},\label{AR/eq-uni2}
	\end{equation}
	where $R_i\in SO(2)$ is the rotation matrix from agent $i$'s local
	coordinate system to the global coordinate system, $u_i^b$
	is the
	controller input applied in agent $i$'s local coordinate system,
	and $z_{ij}^b, z_{ik}^b$ are the local bearings measured in agent $i$'s local
	coordinate system. Since ${(\alpha_{jik}-\alpha_{jik}^*)}$ is a scalar, (\ref{AR/eq-uni2}) and (\ref{AR/eq-uni1}) are equivalent.
\end{remark}
}

\section{Conclusion}
{ 
In this study, we have proposed the angle rigidity theory for the stabilization of planar formations. The notion of angularity has been first defined to describe the multi-point framework with   angle constraints. The established angle rigidity has shown to be a local property because of the existence of flex ambiguity. To check whether an angularity is globally rigid, some sufficient conditions have been proposed. The infinitesimal angle rigidity has been developed based on the trivial motions of the angularity. A sufficient and necessary condition for infinitesimal angle rigidity has been investigated by checking the rank of the angle rigidity matrix. Based on the developed angle rigidity theory, we have also demonstrated how to stabilize a multi-agent planar formation using only angle measurements, which can be realized in each agent's local coordinate system. The exponential convergent rate of angle errors and the collision avoidance between specified agents have  also been proved. Future work will focus on the sufficient and necessary conditions for global angle rigidity and the combinatorial conditions for minimal and infinitesimal angle rigidity.
}

\ifCLASSOPTIONcaptionsoff
\newpage \fi


\bibliographystyle{IEEEtran}
\bibliography{IEEEabrv,Ref_AR}
\end{document}